\documentclass[12pt,draftcls,onecolumn]{IEEEtran}
 \usepackage[update,prepend]{epstopdf}
\usepackage{amsmath,rotating}
\usepackage{hyperref}
\usepackage{amsfonts}
\usepackage{amssymb}
\usepackage{times,subfigure}
\usepackage{comment}
\usepackage{algorithmic}
\usepackage{algorithm}
\usepackage{graphicx}
\usepackage[scanall]{psfrag}

\newtheorem{Proposition}{Proposition}

\newtheorem{Remark}{Remark}

\def\sqw{\hfill\hbox{\lower.1ex\hbox{$\sqcup$}
    \kern-1.02em\lower.1ex\hbox{$\sqcap$}}\ }

\newcommand{\qed}{\hfill \mbox{\raggedright \rule{.07in}{.1in}}}

\newenvironment{proof}{\vspace{1ex}\noindent{\bf Proof}\hspace{0.5em}}
	{\hfill\qed\vspace{1ex}}

\newcommand{\xv}{\mathbf{x}}

\newcommand{\tr}{\operatorname{tr}}
\newcommand{\rank}{\operatorname{rank}}

\begin{document}

\title{Viewing the Welch bound inequality \\from the kernel trick viewpoint}

\author{Liang Dai
\thanks{Liang Dai is a Ph.D. student with the Department of Information Technology, Uppsala University, SE 751 05, Uppsala, Sweden. }
}
\date{}

\maketitle
\begin{abstract} This brief note views to the Welch bound inequality using the idea of the kernel trick from the machine learning research area. From this angle, some novel insights of the inequality are obtained.\end{abstract}

\section{Introduction}
Intuitively speaking the Welch bounds characterizes the lower bound for the inner products of unit vectors in a vector space. This inequality dates back to \cite{c1}. More precisely, given $m$ vectors $\{\xv_1,\xv_2,\cdots,\xv_m\}$ with $\xv_i \in \mathit{C}^{n},\ \|\xv_i\|_2 = 1$ for $i = 1,\cdots,m $, and given integer $p \ge 1$, then it holds that
\begin{align}
\max_{i\ne j}|<\xv_i, \xv_j>| \ge \sqrt[2p]{\frac{1}{m-1}\left[\frac{m}{{n+p-1 \choose p}}-1\right]},
 \label{eq.wel1}
\end{align}
where $<\cdot,\cdot>$ is the standard inner product in $\mathit{C}^{n}$. A more fundamental inequality is the following
\begin{align}
\sum_{i=1}^{m}\sum_{j=1}^{m}|<\xv_i, \xv_j>|^{2p} \ge \frac{m^2}{{n+p-1 \choose p}}.
 \label{eq.wel2}
\end{align}
Since (\ref{eq.wel1}) will be a direct consequence of (\ref{eq.wel2}), in the following, we will refer to (\ref{eq.wel2}) as the Welch bound inequality.

This inequality plays an important role in many research areas, for example, in coding theory for communication and also in compressive sensing theory \cite{c2,c4}. The original derivation for ($\ref{eq.wel2}$) in \cite{c1} is purely analytical, and takes many steps to build the result. In \cite{c4}, the authors gave a fresh geometrical investigation of the inequality for the case when $p=1$. Recently, in \cite{c2}, the authors gave a novel geometric reasoning for the result based on a tensor product argument, which works for arbitrary $p$. In this note, we will study the Welch bound inequality using the idea of the kernel trick, which provides additional insights into the inequality.

The kernel trick is widely used within the field of machine learning. The basic idea behind it is that whenever data enters only in the form of scaler products, this scaler product can be replaced by a different kernel. This in turn opens up for mapping the low dimensional data into high dimensional space, and through the mapping, certain nonlinear structure in the original low dimensional space could be mapped into a linear structure in the high dimensional space \cite{c3}. Usually, this mapping is not explicitly defined, instead it is only defined implicitly through the so-called kernel function. A kernel function, $k(\cdot,\cdot): \mathit{C}^{n}\times\mathit{C}^{n} \rightarrow \mathit{C}$, is a pairwise function, which is also positive semi-definite. Given $k(\cdot,\cdot)$, then there exists a map $\phi(\cdot): \mathit{C}^{n} \rightarrow \mathit{C}^{d}$, where $d$ is the dimension of the mapped space, such that $k(\xv_1,\xv_2) = <\phi(\xv_1),\phi(\xv_2)>$. More discussions will be elaborated in the Remarks in next section. For a through introduction of the kernel trick, please refer to \cite{c3,c31}.

The note is organized as follows. Proposition 1 gives an inequality relating the kernel function. After that, the Remark 1,2,3 will discuss the implications of Proposition 1 to the Welch bound inequality - i.e. the geometric interpretation of the Welch bound inequality from the kernel mapping point of view and one generalization of ($\ref{eq.wel2}$) by choosing another suitable kernel function. Finally, we concludes the note by posing an open question.

\section{Results}
In the following, an inequality for the kernel function is derived.
\begin{Proposition}
Given $m$ vectors $\{\xv_1,\xv_2,\cdots,\xv_m\} \in \mathit{C}^{n}$, and the kernel function $k(\cdot,\cdot): \mathit{C}^{n}\times\mathit{C}^{n} \rightarrow \mathit{C}$ . Define the Gram matrix $G$ as $G_{i,j} = k(\xv_i,\xv_j)$ for $i,j = 1,\cdots, m$, and $r \triangleq \rank(G)$, then we have that
\begin{align}
\sum_{i=1}^{m}\sum_{j=1}^{m} k(\xv_i,\xv_j)^2 \ge \frac{\left(\sum_{l=1}^{m}k(\xv_l,\xv_l)\right)^2}{r}.
\label{eq.wel3}
\end{align}
\end{Proposition}

\begin{proof}
Notice that
\begin{align*}
\sum_{i=1}^{m}\sum_{j=1}^{m} k(\xv_i,\xv_j)^2= \|G\|_{F}^2 = \tr(GG^{H}),
\end{align*}
and
\begin{align*}
\sum_{l=1}^{m}k(\xv_l,\xv_l) = \tr(G).
\end{align*}

Hence, proving (\ref{eq.wel3}) is equivalent to proving that
\begin{align}
\tr(GG^{H}) \ge \frac{\left(\tr(G)\right)^2}{r}.
\label{eq.t1}
\end{align}

The Lemma 2.1 in \cite{c2} implies (\ref{eq.t1}). For readers' convenience, we include its proof as follows.
Denote $\sigma_1, \sigma_2, \cdots, \sigma_r$ as the eigenvalues of $G$, then it follows that
$$\tr(GG^{H}) = \sum_{i=1}^{r} \sigma_i^2 \text{ and } \text{tr}(G) = \sum_{i=1}^{r} \sigma_i.$$
Therefore, proving (\ref{eq.t1}) boils down to prove that
 $$\sum_{i=1}^{r} \sigma_i^2 \ge \frac{\left(\sum_{i=1}^{r} \sigma_i\right)^2}{r},$$
  which holds due to the Cauchy-Schwartz inequality.
\end{proof}

Proposition 1 is valid for any feasible kernel function $k(\cdot,\cdot)$ and arbitrary vectors $\xv_1,\cdots,\xv_m$. The following remarks discuss its implications for the Welch bound inequality and give a new generalization of it by choosing another suitable kernel function.  Note that the facts we used about the polynomial kernel in the following remarks can be found for example in the books \cite{c3,c31}.

\begin{Remark}
Since $k(\xv_i,\xv_j) = \phi(\xv_i)^{H}\phi(\xv_j),$ where $\phi(\cdot): \mathit{C}^{n} \rightarrow \mathit{C}^{d}$ is the implicit map defined by $k(\cdot,\cdot)$, we can decompose the matrix $G$ as $G = D^{H}D$, where $$D = [\phi(\xv_1), \phi(\xv_2), \cdots, \phi(\xv_m)] \in \mathit{C}^{d\times m}.$$ Since the rank of $G$ is $r$ and $G = D^{H}D$, it follows that  $\rank(D) = r$. This fact implies that $r$ can be interpreted as the dimension of the kernel feature space. 

Also note that $r$ could be smaller than $d$, and for example, when the kernel is chosen as $k(\xv_i,\xv_j) = <\xv_i,\xv_j>^{p}$, we have that $d = n^p$ and $r = {n+p-1 \choose p}.$
\end{Remark}

\begin{Remark}
When $k(\xv_i,\xv_j) = <\xv_i,\xv_j>^{p}$, (\ref{eq.wel3}) could be rewritten as
\begin{align}
\frac{\sum_{i=1}^m\sum_{j=1}^m |<\xv_i, \xv_j>|^{2p}}{\left(\sum_{i=1}^m\|\xv_i\|_2^{2p} \right)^2} \ge \frac{1}{{n+p-1 \choose p}},
\end{align}
which is the generalized form of (\ref{eq.wel2}) given in \cite{c5}.

If $\|\xv_i\|_2 = 1$ for $i = 1,\cdots,m$ are further assumed, the original Welch bound inequality (\ref{eq.wel2}) is obtained.
\end{Remark}

\begin{Remark}
When $k(\xv_i,\xv_j) = (<\xv_i,\xv_j> + c)^{p}$, where $c$ is a given constant, then an analog of (\ref{eq.wel3}) will be
\begin{align}
\sum_{i=1}^m\sum_{j=1}^m |<\xv_i, \xv_j>+ c |^{2p} \ge \frac{\left(\sum_{i=1}^m(\|\xv_i\|_2^{2}+c)^p \right)^2}{{n+p \choose p}}.
\end{align}
If $\|\xv_i\|_2 = 1$ holds for $i = 1,\cdots,m$, then it holds that
\begin{align}\label{eq.wel4}
\sum_{i=1}^m\sum_{j=1}^m |<\xv_i, \xv_j>+ c |^{2p} \ge \frac{m^2(1+c)^{2p}}{{n+p \choose p}}.
\end{align}
%
\end{Remark}

\section{Conclusion}
This note builds the link between the Welch bound inequality and the kernel trick in machine learning research. The results are a new geometric interpretation of the inequality as well as a generalization of the inequality. We end up the note by posing the following question: whether there exist other kernel functions (other than the polynomial kernel as used in the original Welch bound inequality, i.e. (\ref{eq.wel2})), such that the kernel feature space also has low rank or approximate low rank property?

\end{document}